\documentclass{stacs_proc}

\usepackage{amsmath}
\usepackage{amssymb}
\usepackage{amsthm}
\usepackage[english]{babel}
\usepackage{verbatim}
\usepackage{pst-node}
\usepackage[numbers,sort&compress]{natbib}
\usepackage[a4paper,textwidth=15cm,textheight=22cm]{geometry}

\newif\ifcomment\commentfalse

\def\commentOFF{\commentfalse}

\long\outer\def\bc#1\ec{{\ifcomment \sloppy  $[${\bf suggest}]
{{#1}} \textbf{[end]} \fi }}

\long\outer\def\br#1\er{{\ifcomment \sloppy  $[${\bf suggest remove}]
{{#1}} \textbf{[end]} \fi }}

\long\outer\def\bo#1\eo{{\ifcomment \sloppy  $[${\bf instead of}]
{\textit{#1}} \textbf{[end]}  \fi }}

\long\outer\def\BC#1\EC{{\ifcomment \sloppy \par \#  \dotfill
{\textsc{#1}} \dotfill \# \par \fi }}

\commentOFF

\newcommand{\sidecomment}[1]{}

\newcommand{\latop}[2]{\genfrac{}{}{0pt}{}{#1}{#2}}

\newcommand{\pc}{\ensuremath{{\bf P}}}
\newcommand{\npc}{\ensuremath{{\bf NP}}}

\newcommand{\alg}{\textsc{ALG}}

\newcommand{\priority}{\textsc{Priority Greedy}}
\newcommand{\fifo}{\textsc{FIFO}}

\newcommand{\ibm}{\textsc{IBM}}
\newcommand{\ibmp}{\textsc{Induced Bipartite Matching}}
\newcommand{\wgp}{\textsc{Wgp}}
\newcommand{\cwgp}{\textsc{C-Wgp}}
\newcommand{\fwgp}{\textsc{F-Wgp}}

\newcommand{\ins}{\mathcal{I}}

\newcommand{\di}{\ensuremath{d_I}}

\newcommand{\RELAT}{\ensuremath{\mathbb{Z}}}

\newcommand{\sfloor}[1]{\big\lfloor #1 \big\rfloor}

\newcommand{\affiliate}[1]{$^{\,#1}$}
\newcommand{\institutenr}[1]{$\!\!\!\!\!^{#1}\,$}

\begin{document}

\title{Minimizing Flow Time in the Wireless Gathering Problem}

\author{V. Bonifaci}{Vincenzo Bonifaci\affiliate{1,3}}
\author{P. Korteweg}{Peter Korteweg\affiliate{2}} 
\author{A. Marchetti-Spaccamela}{Alberto Marchetti-Spaccamela\affiliate{3}}
\author{L. Stougie}{Leen Stougie\affiliate{2,4}}

\address[TUB]{\institutenr{1}Technische Universit\"at Berlin,
Institut f\"ur Mathematik, Berlin, Germany}
  
\address[TUE]{\institutenr{2}Eindhoven University of Technology,
     Dept of Mathematics and Computer Science, Eindhoven, The Netherlands}
\email{p.korteweg@tue.nl, l.stougie@tue.nl}

\address[SAP]{\institutenr{3}University of Rome ``La Sapienza'',
     Dept of Computer and Systems Science, Rome, Italy}
\email{bonifaci@dis.uniroma1.it, alberto@dis.uniroma1.it}

\address[CWI]{\institutenr{4}CWI,
     Amsterdam, The Netherlands}
\email{stougie@cwi.nl}



\begin{abstract}
We address the problem of efficient data gathering in a wireless network through multi-hop communication.
We focus on the objective of minimizing the maximum flow time of a data packet.
We prove that no polynomial time algorithm for this problem can have approximation ratio less than $\Omega(m^{1/3})$ when $m$ packets have to be transmitted, unless $\pc=\npc$.
We then use resource augmentation to assess the performance of a FIFO-like strategy.
We prove that this strategy is 5-speed optimal, i.e., its cost remains within the optimal cost if we allow the algorithm to transmit data at a speed 5 times higher than that of the optimal solution we compare to.
\end{abstract}

\maketitle

\stacsheading{2008}{109-120}{Bordeaux}
\firstpageno{109}

\section{Introduction}
\label{sec:introduction}
Wireless networks are used in many areas of practical interest, such as mobile phone communication, ad-hoc networks, and radio broadcasting. Moreover, recent advances in miniaturization of computing devices equipped with short range radios have given rise to strong interest in sensor networks for their relevance in many practical scenarios (environment control, accident monitoring etc.) \cite{Akyldiz:2002,Pahlavan:1995}.

In many applications of wireless networks data gathering is a critical operation for extracting useful information from the operating environment: information collected  from multiple nodes in the network should be transmitted to a sink that may process the data, or act as a gateway to other networks.
We remark that in the case of wireless sensor networks sensor nodes have limited
computation capabilities, thus implying that data gathering is an even more crucial operation.
For this reasons, data gathering in sensor networks has received significant attention in the last few years; we cite just a few contributions \cite{Akyldiz:2002,Florens:2004}.
The problem finds also applications in Wi-Fi networks when many users need to access a gateway using multi-hop wireless relay-routing \cite{Bermond:2006}.

In this paper we focus on the problem of designing and analysing \emph{simple distributed}  algorithms that have \emph{good approximation guarantees} in \emph{realistic scenarios}.
Namely, we are interested in algorithms that are not only distributed but that are fast and can be implemented with limited overhead: sophisticated algorithms that require solving
complex combinatorial optimization problems are impractical for implementations and have mainly theoretical interest.

In order to formally assess the performance of the proposed algorithms we focus on the minimization of the maximum flow, i.e.\ minimizing the maximum time spent in the system by a packet. Almost all of the  previous literature considered the objective of minimizing the completion time (see for example \cite{Bar-Yehuda:1992, Bar-Yehuda:1993, Bermond:2006, Florens:2004, Gargano:2006, Kumar:2004, Pelc:2002}).
Flow minimization is a largely used criterion in scheduling theory that more suitably allows to assess the quality of service provided when multiple requests occur over time \cite{Chan:2006,Chekuri:2004,Kalya:2000,McCullough:2004}.

The problem of modelling realistic scenarios of wireless sensor networks is complicated by the many parameters that define the communication among nodes and influence the performance of transmissions (see for example \cite{Akyldiz:2002, Schmid:2006}).
In the sequel we assume that stations have a common clock, hence time can be divided into rounds. Each node is equipped with a half-duplex interface; as a result it cannot send and receive during the same round. Typically, not all nodes in the network can communicate with each other directly, hence packets have to be sent through several nodes before they can be gathered at the sink; this is called \emph{multi-hop} routing.

The key issue in our setting is \emph{interference}.
A radio signal has a \emph{transmission} radius, the distance over which the signal is strong enough to send data, and an \emph{interference} radius, the distance over which the radio signal is strong enough to interfere with other radio signals.
If node $i$ is transmitting data to node $j$ we have \emph{interference} (or  \emph{collision}) in communication if $j$ also receives signals from other stations at the same time.
Following Bermond et al.\ \cite{Bermond:2006}, we model the wireless network using a graph and a parameter $\di$. An edge between nodes $i$ and $j$ represents the fact that stations $i$ and $j$ are within transmission range of each other. The parameter $\di$ models the interference radius: a node $j$ successfully receives a signal if one of his neighbors is transmitting, and no other node within distance $\di$ from $j$ is transmitting in the same round.
The case $\di=1$ has been extensively considered earlier (see e.g.\ \cite{Bar-Yehuda:1993,Florens:2004,Gargano:2006}); but we remark that assuming $\di=1$ or assuming that interferences/transmissions are modeled according to the well known unit disk graph model \emph{does not} adequately represent interferences as they occur in practice \cite{Schmid:2006}.

Kumar et al.~\cite{Kumar:2005} give an overview of other interference models, including the so-called \emph{distance-2 interference model}.
The distance-2 interference model is similar to our interference model with $\di=1$, plus the extra constraint that no two transmitting nodes should be adjacent;
we observe that this requirement might pose unnecessary conditions.

\medskip
\textbf{The Wireless Gathering Problem.}
An instance of the {\em Wireless Gathering Problem} (\wgp) is given by a static wireless network which consists of several stations (nodes) and one base station (the sink), modeled as a graph, together with the interference radius $\di$; over time data packets arrive at stations that have to be gathered at the base station.

A feasible solution of an instance of \wgp\ is a schedule without interference which determines for each packet both route and times at which it is sent. As we will see in Section~\ref{lb_math} this can be modeled as a clean combinatorial optimization problem. We consider two different objectives.
One is to minimize completion time, i.e., the time needed to gather all the packets. Another, perhaps more natural, objective is minimization of maximum flow time of packets, i.e., the maximum difference between
release time and arrival time at the sink of a packet.
We call these two problems \cwgp\ and \fwgp, respectively.

\medskip
\textbf{Related work.}
The Wireless Gathering Problem was introduced by \citet{Bermond:2006} in the context of wireless access to the Internet in villages. The authors proved that the problem of minimizing the completion time is \npc-hard and presented a greedy algorithm with asymptotic approximation ratio at most $4$.
In \cite{Bonifaci:2006} we considered arbitrary release times and proposed an on-line greedy algorithm with the same approximation ratio.

\citet{Bar-Yehuda:1993} considered distributed algorithms for \cwgp.
Their model is a special case of our model, where $\di=1$ and there are no release dates.
Kumar et al.~\cite{Kumar:2004} studied the more general \emph{end-to-end} transmission problem, where each of the packets may have a different source \emph{and} destination in the network. Under the assumption of a distance-2 interference model, Kumar et al.\ considered the objective of minimizing the maximum completion time of the schedule, and presented hardness results and approximation algorithms for arbitrary graphs and disk graphs.
They presented distributed algorithms for packet scheduling over fixed routing paths, and used a linear program in order to determine the paths.
By contrast, we use a shortest paths tree to fix the routing paths, which is easier to implement in a distributed setting.

\citet{Florens:2004} considered the minimization of the completion time of data gathering in a setting with unidirectional antennas. They provided a 2-approximation algorithm for tree networks and an optimal algorithm for line networks.
\citet{Gargano:2006} gave a polynomial time algorithm for the special case of the same model  in which each node has exactly one packet to send.

Another related problem is to compute the throughput of a wireless network. This has been studied for example in \cite{Kumar:2005}.
We also observe that many papers study broadcasting in wireless networks \cite{Bar-Yehuda:1992,Pelc:2002}.
However, we stress that data gathering and broadcasting are substantially different tasks in the context of packet networks.
In particular, the idea of reversing a broadcast schedule to obtain a gathering schedule (which works when data can be aggregated) cannot be used.

\medskip
\textbf{Main results.}
In Section \ref{sec:hardness} we give inapproximability results for \fwgp.
We prove that \fwgp\ on $m$ packets cannot be approximated within $\Omega(m^{1/3})$, unless $\pc=\npc$. We also show that any algorithm using shortest paths in order to route the packets to the sink is no better than an $\Omega(m)$-approximation.

In Section \ref{sec:algorithms} we present a polynomial time resource augmented approximation algorithm for \fwgp\, which is in fact an on-line algorithm. We use resource augmentation because \fwgp\ is hard to approximate within a reasonable factor.

Resource augmentation was introduced in the context of machine scheduling in \cite{Kalya:2000}: the idea is to study the performance of on-line algorithms which are given processors faster than the adversary. Intuitively, this has been done to compensate an on-line scheduler for its lack of future information. Such an approach has led to a number of interesting results showing that moderately faster processors are sufficient to attain satisfactory performance guarantee
for various scheduling problems, e.g. \cite{Chekuri:2004,Kalya:2000}

Surprisingly, in the case of \fwgp\ a modest resource augmentation allows to compensate not only the lack of future information but also the approximation hardness of the problem: we present a $\sigma$-speed optimal algorithm for \fwgp\ and \cwgp; $\sigma$ depends on $\di$ and is at most $5$.

We remark that our algorithm can be implemented using local information only:
in particular, it suffices that a node is informed about the state of nodes within distance $\di+1$. On the other hand, our lower bounds hold for centralized algorithms as well.

\section{Mathematical preliminaries}
\label{lb_math}
We formulate \wgp\ as a graph optimization problem. The model we use can be seen as a generalization of a well-known model for packet radio networks \cite{Bar-Yehuda:1992,Bar-Yehuda:1993}.

An instance of \wgp\ consists of a graph $G=(V,E)$, a \emph{sink} node $s \in V$, a positive integer $\di$, and a set of \emph{data packets} $J=\{1,2,\ldots,m\}$.
Each packet $j \in J$ has an \emph{origin} $o_j \in V$ and a \emph{release date} $r_j \in \RELAT_+$.

We assume that time is discrete; we call a time unit a \emph{round}.
The rounds are numbered $0,1,2,\ldots$. During each round a node may either be \emph{sending} a packet, be \emph{receiving} a packet or be \emph{inactive}.
If two nodes $u$ and $v$ are adjacent, then $u$ can send a packet to $v$ during a round.
If node $u$ sends a packet $j$ to $v$ in some round, then the pair $(u,v)$ is said to be a \emph{call} from $u$ to $v$.
For each pair of nodes $u,v \in V$, the \emph{distance} between $u$ and $v$, denoted by $d(u,v)$, is the length of a shortest path from $u$ to $v$ in $G$.
Two calls $(u,v)$ and $(u',v')$ \emph{interfere} if they occur in the same round and either $d(u',v) \le \di$ or $d(u,v') \le \di$; otherwise the calls are \emph{compatible}. For this reason, the parameter $\di$ is called the \emph{interference radius}. The special case of a unit interference radius corresponds to the above cited model of \citet{Bar-Yehuda:1992}.

For every packet $j \in J$, the release date $r_j$ specifies the time at which the packet enters the network, i.e.\ packet $j$ cannot be sent before round $r_j$.
In the off-line version the entire instance is completely known at time 0; in the on-line version information about a packet becomes known only at its release date.

A solution for a \wgp\ instance is a schedule of compatible calls such that all packets are ultimately sent to the sink. Notice that while in principle each radio transmission can broadcast the same packet to multiple destinations, in the gathering problem having more than one copy of the same packet does not help, as it suffices to keep the one that will arrive first at the sink. Thus, we assume that at any time there is a unique copy of each packet.  Also, in the model we consider, packets cannot be aggregated.

Given a schedule, let $x_j^t$ be the unique node holding packet $j$ at time $t$. The integer
$C_j := \min \{ t: x_j^t=s \}$ is called the \emph{completion time} of packet $j$, while $F_j := C_j-r_j$ is the \emph{flow time} of packet $j$.
In this paper we are interested in the minimization of $\max_j F_j$ (\fwgp).
As an intermediate step in the analysis of \fwgp, we also study the minimization of $\max_j C_j$ (\cwgp).

Some auxiliary notation, we denote by $\delta_j := d(o_j,s)$ the minimum number of calls required for packet $j$ to reach $s$. We also define $\gamma := \di+2$, and $\gamma_0 := \sfloor{(\di+1)/2}$.

We analyze the performance of our algorithms using the standard worst case analysis techniques of approximation ratio analysis, as well as resource augmentation. Given a \wgp\ instance $\ins$ and an algorithm $\alg$, we define $\mathcal{C(\ins)}$ as the cost of \alg\ and $\mathcal{C}^*(\ins)$ as the cost of the optimal solution on $\ins$.
A polynomial-time algorithm is called an $\alpha$-approximation if for any instance $\ins$ we have $\mathcal{C(\ins)} \le \alpha \cdot \mathcal{C^*}(\ins)$.

In the resource augmentation paradigm, the algorithm is allowed to use more resources than the adversary. We consider resource augmentation based on speed, in which we assume that the algorithm can schedule compatible calls with higher speed than the optimal algorithm.
For any $\sigma \ge 1$, we call an algorithm a $\sigma$-speed algorithm if the time used by the algorithm to schedule a set of compatible calls is $1/\sigma$ time units. See \cite{Ausiello:1999} for more information on approximation algorithms, and \cite{Kalya:2000} for more on resource augmentation.

\section{Inapproximability}
\label{sec:hardness}
In this section we prove an inapproximability result for \fwgp.
To prove this result we consider the so-called \emph{induced matching} problem.
A matching $M$ in a graph $G$ is an \emph{induced matching} if no two edges in $M$ are joined by an edge of $G$.
The following rather straightforward relation between compatible calls in a bipartite graph and induced matchings will be crucial in the following.
\begin{proposition}
\label{lem:matching}
Let $G=(U,V,E)$ be a bipartite graph with node sets $(U,V)$ and edge set $E$. Then, a set
$M \subseteq E$ is an induced matching if and only if the calls corresponding to edges of $M$, directed from $U$ to $V$, are all pairwise compatible, assuming $\di=1$.
\qed
\end{proposition}

\noindent
\ibmp\ (\ibm)\\
Instance: a bipartite graph $G$ and an integer $k$. \\
Question: does $G$ have an induced matching of size at least $k$?
\smallskip

We will use the fact that the optimization version of \ibm\ is hard to approximate: there exists an $\alpha >1$ such that it is $\npc$-hard to distinguish between graphs with induced matchings of size $k$ and graphs in which all induced matchings are of size at most $k/\alpha$.
The current best bound for $\alpha$ is $6600/6599$ \cite{Duckworth:2005}.

\begin{theorem}
\label{thm:apx}
Unless $\pc=\npc$, no polynomial-time algorithm can approximate \fwgp\
within a ratio better than $\Omega(m^{1/3})$.
\end{theorem}
\begin{proof}
Let $(G,k)$ be an instance of \ibm, $G=(U,V,E)$. We construct a 4-layer network with a unique source $o$ (first layer), a clique on $U$ and a clique on $V$ (middle layers), and a sink $s$ (last layer). Source $o$ is adjacent to each node in $U$, and $s$ to each node in $V$. The edges between $U$ and $V$ are the same as in $G$ (see Figure \ref{fig:reduction}). We set $\di=1$.

\begin{figure}
\begin{center}
\psset{unit=0.25cm,radius=0.5,fillstyle=solid,fillcolor=white}
\newpsobject{showgrid}{psgrid}{subgriddiv=1,griddots=10,gridlabels=0pt}
\begin{pspicture}(-6,-4.25)(6,4.25)
\Cnode(0,5){s}\nput{0}{s}{$o$}
\Cnode(0,-5){t}\nput{0}{t}{$s$}
\Cnode(-6,2){p1}
\Cnode(6,2){p7}
\Cnode(-6,-2){q1}
\Cnode(6,-2){q7}
\ncline{p1}{p7}\ncline{q1}{q7}\ncline{p1}{q1}\ncline{p7}{q7}
\ncbox[linewidth=0.1,linearc=0.9,boxheight=1,boxdepth=1,fillcolor=lightgray,nodesep=0.5]{p1}{p7}
\ncbox[linewidth=0.1,linearc=0.9,boxheight=1,boxdepth=1,fillcolor=lightgray,nodesep=0.5]{q1}{q7}
\Cnode(-6,2){x1}
\Cnode(-4,2){x2}
\Cnode(-2,2){x3}
\Cnode(0,2){x4}
\Cnode(2,2){x5}
\Cnode(4,2){x6}
\Cnode(6,2){x7}
\Cnode(-6,-2){y1}
\Cnode(-4,-2){y2}
\Cnode(-2,-2){y3}
\Cnode(0,-2){y4}
\Cnode(2,-2){y5}
\Cnode(4,-2){y6}
\Cnode(6,-2){y7}
\ncline{s}{x1}\ncline{s}{x2}\ncline{s}{x3}\ncline{s}{x4}\ncline{s}{x5}\ncline{s}{x6}\ncline{s}{x7}
\ncline{t}{y1}\ncline{t}{y2}\ncline{t}{y3}\ncline{t}{y4}\ncline{t}{y5}\ncline{t}{y6}\ncline{t}{y7}
\rput(0,0){$G$}
\rput(8,2){$U$}
\rput(8,-2){$V$}
\end{pspicture}
\caption{The construction in the proof of Theorem \ref{thm:apx}}
\label{fig:reduction}
\end{center}
\end{figure}
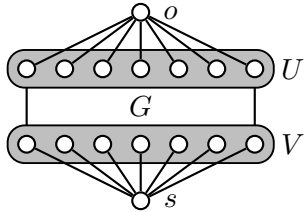

The \fwgp\ instance consists of $m:=(1-1/\alpha)^{-1} (1+k/\alpha) (2k+1) k = \Theta(k^3)$ packets with origin $o$. They are divided into $m/k$ groups of size $k$. Each packet in the $h$th group has release date $(k+1)h$, $h=0,\ldots,m/k-1$.
Rounds $(k+1)h$ till $(k+1)(h+1)-1$ together are a \emph{phase}.

We prove that if $G$ has an induced matching of size $k$, there is a gathering schedule of cost $2k+1$, while if $G$ has no induced matching of size more than $k/\alpha$, every schedule has cost at least $(2k+1)k = (2k+1) \Theta(m^{1/3})$. The theorem then follows directly.

Assume $G$ has an induced matching $M$ of size $k$, say $(u_i,v_i)$, $i=0\ldots k-1$. Then consider the following gathering schedule. In each phase, the $k$ new packets at $o$ are transmitted, necessarily one-by-one, to layer $U$ while old packets at layer $V$ (if any) are absorbed at the sink; then, in a \emph{single} round, the $k$ new packets move from $U$ to $V$ via the matching edges. More precisely, each phase can be scheduled in $k+1$ rounds as follows:

\noindent
1. for $i=0,\ldots,k-1$ execute in the $i$th round the two calls $(o,u_{i})$ and $(v_{i+1 \bmod k},s)$;\\
2. in the $k$th round, execute simultaneously all the calls $(u_i,v_i)$, $i=0,\ldots,k-1$.

\noindent
The maximum flow time of the schedule is $2k+1$, as a packet released in phase $h$ reaches the sink before the end of phase $h+1$.

In the other direction, assume that each induced matching of $G$ is of size at most $k/\alpha$. By Proposition \ref{lem:matching}, at most $k/\alpha$ calls can be scheduled in any round from layer $U$ to layer $V$. We ignore potential interference between calls from $o$ to $U$ and calls from $V$ to $s$; doing so may only decrease the cost of a schedule. As a consequence, we can assume that each packet follows a shortest path from $o$ to $s$.
Notice however that, due to the cliques on the layers $U$ and $V$, no call from $U$ to $V$ is compatible with a call from $o$ to $U$, or with a call from $V$ to $s$.

Let $m_o$ and $m_U$ be the number of packets at $o$ and $U$, respectively, at the beginning of a given phase. Also, let $\beta := 1+k/\alpha$.
We associate to the phase a potential value $\psi := \beta m_o + m_U$, and we show that at the end of the phase the potential will have increased proportionally to $k$. Let $c_o$ and $c_U$ denote the number of calls from $o$ to $U$ and from $U$ to $V$, respectively, during the phase. Since a phase consists of $k+1$ rounds, and in each round at most $k/\alpha$ calls are scheduled from $U$ to $V$, we have $c_o + c_U/(k/\alpha) \leq k+1$, or, equivalently  since $k/\alpha = \beta-1$,
\begin{equation}
\label{eq:rate-bound}
(\beta-1) c_o + c_U \leq (\beta-1)(k+1).
\end{equation}
If $m'_o$, $m'_U$ are the number of packets at $o$ and $U$ at the
beginning of the next phase, and $ \psi'=\beta m'_o + m'_U $ is the new potential, we have
\begin{eqnarray*}
m'_o &=& m_o + k - c_o \\
m'_U &=& m_U + c_o - c_U \\
\psi' -\psi &=& \beta (m'_o-m_o) + m'_U-m_U \\
&=& \beta (k-c_o) + c_o - c_U \\
&=& \beta k - (\beta-1) c_o - c_U \\
&\geq& \beta k - (\beta-1)(k+1)  \\
&=& k - (\beta-1) \\
&=& (1-1/\alpha) k
\end{eqnarray*}
where the inequality uses \eqref{eq:rate-bound}.

Thus, consider the situation after $m/k$ phases. The potential has become at least $\Psi := (1-1/\alpha)m$. By definition of the potential, this implies that at least $\Psi/\beta = (1-1/\alpha)(1+k/\alpha)^{-1} m = (2k+1)k$ packets reside at either $o$ or $U$; in particular, they have been released but not yet absorbed at the sink. Since the sink cannot receive more than one packet per round, this clearly implies a maximum flow time of $(2k+1)k = (2k+1) \Theta(m^{1/3})$ for one of these packets.
\end{proof}


In cases where the packets are routed via shortest paths to the sink -- a behavior common to many gathering protocols -- the result of
Theorem \ref{thm:apx} can be strengthened further.

\begin{theorem}
\label{thm:shortest-paths}
No algorithm that routes packets along shortest paths can approximate \fwgp\ within a ratio better than $\Omega(m)$.
\end{theorem}

\begin{proof} Consider the instance in Figure \ref{fig:lb-sp}.
The adversary releases a message at each of the nodes $u_1,u_2,u_3$
at times $5i$, $i=0,\ldots,m/3$.
Any shortest paths following algorithm sends all messages via $u$, yielding
$\max_j C_j \geq 3m$. As $r_j \le 5m/3$ for each message $j$, we have
$\max_j F_j \ge 3m - 5m/3 = 4m/3$.

The adversary sends each message over the path which does not contain $u$.
We claim that it is possible to do this so that all messages released
at time $5i$ arrive at the sink in round $5(i+1)+1$ latest.
If the claim holds, then we have $\max_j F^*_j \le 5(i+1)+1-5i=6$, from which the theorem will follow.

We prove the claim by induction.
Suppose the claim holds for messages released in round $5(i-1)$.
Then, the last message released at time $5(i-1)$ latest is sent to the sink in round $5i$. This message does not block any message released in round $5i$.
Now, the adversary sends the messages released in round $5i$ to a node adjacent to $s$ in 3 rounds, i.e. in the rounds $5i$, $5i+1$ and $5i+2$.
Then, it requires another 3 rounds to send all 3 messages to the sink, i.e.\ the rounds $5i+3$, $5i+4$, and $5(i+1)$. This proves the theorem, since $\max_j F_j / \max_j F^*_j \geq (4m/3)/6  = 2m/9$.
\end{proof}

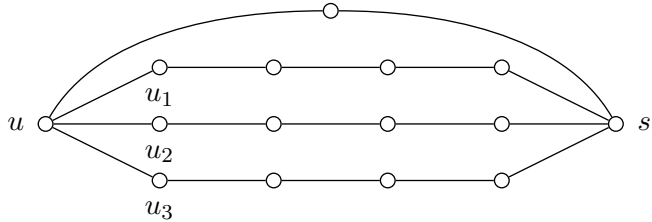
\begin{figure}[!ht]
\begin{center}
\psset{unit=0.75cm,arrows=-,linewidth=0.5pt,radius=3pt}
\begin{pspicture}(-8,-1.5)(10,1.75)
\Cnode(0,2){q2}
\Cnode(-5,0){q1}
\Cnode(-3,0){n21}
\Cnode(-1,0){n22}
\Cnode(1,0){n23}
\Cnode(3,0){n24}
\Cnode(5,0){s}
\Cnode(-3,1){n11}
\Cnode(-1,1){n12}
\Cnode(1,1){n13}
\Cnode(3,1){n14}
\Cnode(-3,-1){n31}
\Cnode(-1,-1){n32}
\Cnode(1,-1){n33}
\Cnode(3,-1){n34}
\nput{0}{s}{$s$}
\nput{180}{q1}{$u$}
\nput{-90}{n11}{$u_1$}
\nput{-90}{n21}{$u_2$}
\nput{-90}{n31}{$u_3$}
\ncline{n11}{n12}
\ncline{n12}{n13}
\ncline{n13}{n14}
\ncline{n21}{n22}
\ncline{n22}{n23}
\ncline{n23}{n24}
\ncline{n31}{n32}
\ncline{n32}{n33}
\ncline{n33}{n34}
\nccurve[angleA=60,angleB=180]{q1}{q2}
\nccurve[angleA=120,angleB=0]{s}{q2}
\ncline{n11}{q1}
\ncline{n21}{q1}
\ncline{n31}{q1}
\ncline{n14}{s}
\ncline{n24}{s}
\ncline{n34}{s}
\end{pspicture}
\caption{No shortest path based algorithm is better than $\Omega(m)$-approximate ($\di=1$).}
\label{fig:lb-sp}
\end{center}
\end{figure}

\section{Approximation Algorithms}
\label{sec:algorithms}

In this section we present and analyze a \fifo\ algorithm for \wgp. First, we show that \fifo\ is a $5$-approximation for \cwgp. Note that the best approximation algorithm known is $4$-approximate; the main interest in analyzing \fifo\ is that we use it as a subroutine in an algorithm for \fwgp\ which uses resource augmentation. Next, we prove that this algorithm with resource augmentation is a $\sigma$-speed optimal algorithm, for any $\sigma \ge 5$, for both \cwgp\ and \fwgp.

\subsection{An approximation algorithm for C-WGP}
We will present an approximation algorithm for \cwgp.
The algorithm we consider is actually a special case of a general scheme for which we can prove an upper bound on the completion time \cite{Bonifaci:2006}. In this scheme, called \priority, each packet is assigned a unique priority based on some algorithm-specific rules. Then, in each round, packets are considered in order of decreasing priority and are sent towards the sink as long as there is no interference with higher priority packets.

\begin{algo}[\priority]
In every round, consider the available packets in order of decreasing priority, and send each next packet along a shortest path from its current node to $s$, as long as this causes no  interference with any higher-priority packet.
\end{algo}

\noindent
We first derive upper bounds on the completion time $C_j$ of each packet $j$ in a \priority\ solution.

We say that packet $j$ is \emph{blocked} in round $t$ if $t \geq r_j$ but $j$ is not sent in round $t$.
Note that in a \priority\ algorithm a packet can only be blocked due to interference with
a higher priority packet. We define the following \emph{blocking relation} on a \priority\ schedule: $k \prec j$ if in the last round in which $j$ is blocked,
$k$ is the packet closest to $j$ that is sent in that round and has a priority higher than $j$ (ties broken arbitrarily).
The blocking relation induces a directed graph $F=(J,A)$ on the packet set $J$ with an arc $(k,j)$ for each $k,j \in J$ such that $k \prec j$.
Observe that, for any \priority\ schedule, $F$ is a directed forest and the root of each tree of $F$ is a packet which is never blocked.
For each $j$ let $T(j) \subseteq F$ be the tree of $F$ containing $j$,
$b(j)\in J$ be the root of $T(j)$, and $P(j)$ the set of packets along the path in $F$ from $b(j)$ to $j$.
Finally, define
$\pi_j := \min \{ \delta_j, \gamma_0 \}$ and $R_j := r_j + \delta_j - \pi_j$.

We have upper and lower bounds on the completion time of a packet.
\begin{lemma}[\cite{Bonifaci:2006}]
\label{lemma:ub}
For each packet $j \in J$,
$ C_j \leq R_{b(j)} + (\gamma/\gamma_0) \cdot \sum_{i \in P(j)} \pi_i. $
\end{lemma}
\begin{lemma}[\cite{Bonifaci:2006}]
\label{lemma:lb}
Let $S \subseteq J$ be a nonempty set of packets, and let $C^*_i$ denote the completion time of packet $i$ in some feasible schedule. Then there is $k \in S$ such that
$ \max_{i \in S} C^*_i \geq R_k + \sum_{i \in S} \pi_i. $
\end{lemma}

Our algorithm is based on a version of the \priority\ scheme, in which a
higher priority is given to packets with earlier release dates (ties broken arbitrarily).
We call this algorithm \fifo\ after the famous {\em first-in-first-out} algorithm
in scheduling and service systems, though in our case packets do not necessarily arrive
in order of their priority at the sink.

\begin{theorem}
\label{thm:fifo}
\fifo\ is a $(1+\gamma/\gamma_0)$-approximation algorithm for \cwgp.
\end{theorem}

\begin{proof}
Let $j$ be the packet having maximum $C_j$, and consider $T(j)$, the tree containing $j$ in the forest induced by the blocking relation. We can apply Lemma \ref{lemma:lb} with $S=T(j)$ to obtain
\begin{eqnarray}
\label{eqn:fifo:lb}
\max_{i \in T(j)} C^*_i \geq r_k + \delta_k + \sum_{\latop{i \in T(j)}{i \neq k}} \pi_i
\end{eqnarray}
where $k$ is some packet in $T(j)$.
On the other hand, by using Lemma \ref{lemma:ub},
\begin{eqnarray}
\label{eqn:fifo:ub}
C_j & \leq & R_{b(j)} + \frac{\gamma}{\gamma_0} \sum_{i \in P(j)} \pi_i \\
& = & r_{b(j)} + \delta_{b(j)} - \pi_{b(j)} + \frac{\gamma}{\gamma_0} \sum_{i \in P(j)} \pi_i \nonumber \\
& \le & r_{b(j)} + \frac{\gamma}{\gamma_0} \min\{\delta_k, \gamma_0\} + \frac{\gamma}{\gamma_0} \sum_{\latop{i \in P(j)}{i \neq k}} \pi_i
+ \delta_{b(j)} \nonumber \\
& \le & \frac{\gamma}{\gamma_0} \bigg( r_k + \delta_k + \sum_{\latop{i \in T(j)}{i \neq k}} \pi_i \bigg) + \delta_{b(j)}. \nonumber
\end{eqnarray}
\BC Peter: removed "min" from lines 2 and 3, and replaced = with $\le$ in line 3. V: OK.\EC
where we used the fact that, by definition of \fifo, 
we have $r_{b(j)} \leq r_k$.
Equations \eqref{eqn:fifo:lb} and \eqref{eqn:fifo:ub}, and observation $\max_{i \in T(j)} C^*_i \geq \delta_{b(j)}$ prove the theorem.
\end{proof}

It is straightforward to verify that $2 \leq \gamma/\gamma_0 \leq 4$ for all $\di$, while $\gamma/\gamma_0 = 3$ for $\di=1$.
\begin{corollary}
\fifo\ is a $5$-approximation algorithm for \cwgp. When $\di=1$, \fifo\ is a $4$-approximation for \cwgp.
\end{corollary}

The bound on the approximation ratio of \fifo\ is slightly worse than that
of a \priority\ algorithm based on $R_j$, which is a $\gamma/\gamma_0$-approximation.
In fact, we also have an example on which \fifo\ is strictly worse than a $\gamma/\gamma_0$-approximation (we omit the example here due to space limitations). However, we remark that \fifo\ is both natural and simple; and, perhaps more importantly, Theorem \ref{thm:fifo} will be instrumental
in proving good bounds for the minimization of maximum flow time, where we will use \fifo\ as a subroutine of our algorithm.

\subsection{A resource augmentation bound for F-WGP}
Motivated by the hardness result of Section \ref{sec:hardness}, we study algorithms under resource augmentation. In this context we study $\sigma$-speed algorithms, in which data packets are sent at a speed that is $\sigma$ times faster than the solution we compare to.

\begin{algo}[$\sigma$-\fifo] \ \\
1. Create a new instance $\ins'$ by multiplying release dates: $r'_j := \sigma r_j$; \\
2. Run \fifo\ on $\ins'$; \\
3. Speed up the schedule thus obtained by a factor of $\sigma$.
\end{algo}
The schedule constructed by $\sigma$-\fifo\ is a feasible $\sigma$-speed solution to the original problem because of step 1.
We will show that $\sigma$-\fifo\ is optimal for both \cwgp\ and \fwgp, if $\sigma \ge \gamma/\gamma_0+1$. The following Lemma is crucial.
\begin{lemma}
\label{lem:speed_makespan_to_flow}
If $\sigma$-\fifo\ is a $\sigma$-speed optimal algorithm for \cwgp, then it is also a $\sigma$-speed optimal algorithm for \fwgp.
\end{lemma}

\begin{proof}
Let $F^*_j$ and $F_{j,\sigma}$ be the flow time of data packet $j$ in an optimal solution and in a $\sigma$-\fifo\ solution, respectively, to \fwgp\ and let $C^*_j$ and $C_{j,\sigma}$ be the completion time of data packet $j$ in the same solutions.
Suppose $\sigma$-\fifo\ is a $\sigma$-speed optimal algorithm for \cwgp,
hence we have $\max_{j \in J} C_{j,\sigma} \le \max_{j \in J} C^*_j$.
We show that this inequality implies, for any time $t$,
\begin{eqnarray}
\label{eq:flow}
\max_{{j \in J},\, {r_j = t}}C_{j,\sigma} \le \max_{{j \in J},\, {r_j \le t}} C^*_j.
\end{eqnarray}
We prove inequality \eqref{eq:flow} by contradiction.
Suppose it is false, then there is an instance $\ins$ of minimum size (number of data packets) for which it is false.
Also, let $t_0$ be the first round in such an instance for which it is false.
By definition, $\sigma$-\fifo\ schedules each data packet $j$ definitively in round $r_j$; no data packet is rescheduled in a later round.
I.e., the algorithm determines the completion time $C_{j,\sigma}$ of data packet $j$ in round $r_j$. If the inequality is false, then we must have
\begin{eqnarray}
\label{eq:flow2}
C_{i,\sigma} > \max_{j \in J,\, r_j \le t_0} C^*_j,
\end{eqnarray}
for some data packet $i$ with $r_i = t_0$, and because $\ins$ is a minimum size instance the instance does not contain any data packets released after round $t_0$. But then \eqref{eq:flow2} contradicts
$\max_{j \in J} C_{j,\sigma} \le \max_{j \in J} C^*_j$ .
Using \eqref{eq:flow} we have
\begin{eqnarray*}
\max_{j \in J} F_{j,\sigma} &=& \max_t \bigg( \max_{j \in J,\, r_j=t}C_{j,\sigma}-t \bigg) \le \max_t \bigg( \max_{j \in J,\, r_j\le t} C^*_j-t \bigg) \\
& \le & \max_t \bigg( \max_{j \in J,\, r_j\le t}F^*_j \bigg) = \max_{j \in J} F^*_j.
\end{eqnarray*}
\end{proof}

\begin{theorem}
For $\sigma \ge \gamma/\gamma_0 +1$, $\sigma$-\fifo\ is a $\sigma$-speed optimal algorithm for both \cwgp\ and \fwgp.
\end{theorem}
\begin{proof}
By Lemma \ref{lem:speed_makespan_to_flow}, it suffices to prove that $\sigma$-\fifo\ is $\sigma$-speed optimal for \cwgp.

Let $C_j$ be the completion time of any data packet $j$ in the $\sigma$-\fifo\ solution on instance $\ins$, and let $C'_j$ be the completion time of $j$ in the \fifo\ solution on the instance $\ins'$ (see the algorithm description). By construction $C_j=C'_j/\sigma$. Let $R'_j := \sigma r_j + \delta_j - \pi_j$. 
Then the upper bound of Lemma \ref{lemma:ub} applied to instance $\ins'$ implies $C'_j \leq R'_{b(j)} + (\sigma-1) \sum_{i \in P(j)} \pi_i$.
Hence,
\begin{eqnarray}
\label{eq:cjub}
C_j = C'_j / \sigma \leq \frac{1}{\sigma} R'_{b(j)} + \frac{\sigma - 1}{\sigma} \sum_{i \in P(j)} \pi_i
& \leq & r_{b(j)} + \frac{1}{\sigma} \delta_{b(j)} + \frac{\sigma-1}{\sigma} \sum_{i \in P(j)} \pi_i.
\end{eqnarray}
Since in any solution $b(j)$ has to reach the sink we clearly have
\begin{equation}
\label{eq:simplelb}
\max_{i \in P(j)} C^*_i \geq C^*_{b(j)} \geq r_{b(j)} + \delta_{b(j)}.
\end{equation}
Also, by Lemma \ref{lemma:lb}, for some $k \in P(j)$,
\begin{eqnarray}
\label{eq:complexlb}
\max_{i \in P(j)} C^*_i \geq R_k + \sum_{i \in P(j)} \pi_i
\geq r_k + \sum_{i \in P(j)} \pi_i  \geq r_{b(j)} + \sum_{i \in P(j)} \pi_i,
\end{eqnarray}
where the last inequality follows from $b(j)$ having lowest release time in $P(j)$, by definition of \fifo.
Combining \eqref{eq:cjub}, \eqref{eq:simplelb} and \eqref{eq:complexlb}, we obtain
\begin{eqnarray*}
\max_{i \in P(j)} C^*_i &=& \frac{1}{\sigma} \max_{i \in P(j)} C^*_i + \frac{\sigma-1}{\sigma} \max_{i \in P(j)} C^*_i \\
& \geq & \frac{1}{\sigma} \bigg( r_{b(j)} + \delta_{b(j)} \bigg) + \frac{\sigma-1}{\sigma}  \bigg( r_{b(j)} + \sum_{i \in P(j)} \pi_i \bigg) \\
& = & r_{b(j)} + \frac{1}{\sigma} \delta_{b(j)} + \frac{\sigma-1}{\sigma} \sum_{i \in P(j)} \pi_i \geq C_j. 
\end{eqnarray*}
\end{proof}

\begin{corollary} $5$-\fifo\ is a $5$-speed optimal algorithm for \cwgp\ and \fwgp.
\end{corollary}

\subsection{Another upper bound for FIFO}

As we have seen in Section \ref{sec:hardness}, \fwgp\ is extremely hard to approximate without resource augmentation -- no bound better than $\Omega(m^{1/3})$ is possible. Moreover, algorithms that route along shortest paths cannot do better than $\Omega(m)$ (recall Theorem \ref{thm:shortest-paths}). In this section we show that \fifo\ is in fact an $O(m)$-approximation for \fwgp. Thus, apart from constant factors, \fifo\ is best possible among algorithms that use shortest paths.

\begin{theorem}
\label{thm:fifo:lb}
\fifo\ is an $O(m)$-approximation for \fwgp.
\end{theorem}
\begin{proof}
Since every packet must be gathered at the sink, clearly $\max_j F^*_j \geq \max_j \delta_j \geq \max_j \pi_j$.
Now let $j$ be the packet incurring the maximum flow time in the schedule obtained by \fifo.
Since $r_j \geq r_{b(j)}$ (by definition of \fifo), we have
\begin{equation}
\label{eq:rj}
R_{b(j)} - r_j = r_{b(j)} + \delta_{b(j)} - \pi_{b(j)} - r_j \leq \delta_{b(j)}
\end{equation}

Using Lemma \ref{lemma:ub} and \eqref{eq:rj}, we get
\begin{eqnarray*}
F_j = C_j - r_j & \leq & R_{b(j)} - r_j + \frac{\gamma}{\gamma_0} \sum_{i \in P(j)} \pi_i \\
& \leq & \delta_{b(j)} + \frac{\gamma}{\gamma_0} \sum_{i \in P(j)} \pi_i \\
& \leq & \max_{i} F^*_i + \frac{\gamma}{\gamma_0} \cdot |P(j)| \cdot \max_i F^*_i \\
& \leq & \left( 1 + \frac{\gamma}{\gamma_0} m \right) \max_i F^*_i.
\end{eqnarray*}
\end{proof}

\section{Conclusion}
\label{lb_conclusion}
We considered the wireless gathering problem with the objective of minimizing the maximum flow time of data packets (\fwgp). We showed that the simple on-line algorithm \fifo\ has favorable behavior: although the problem is extremely hard to approximate in general, augmenting the transmission rate by a factor of $5$ allows \fifo\ to remain within the cost of an optimal solution for the problem without augmentation.

It is an open question whether optimality can be achieved by augmenting the transmission rate by a factor smaller than $5$, and whether an efficient algorithm exists that matches the $\Omega(m^{1/3})$ lower bound on the approximability of \fwgp.

Another interesting set of questions concerns resource augmentation by allowing the algorithms to use extra frequencies, meaning that more than one data packet can be sent simultaneously over the same channel. For instance, does there exist a 5-frequency optimal FIFO-type algorithm?

For the minimization of the completion time (\cwgp), the existence of a polynomial time approximation scheme is still open.
It is known that no algorithm that uses shortest paths to route the data packets to the sink can give an improvement over the currently best approximation ratio \cite{Bonifaci:2006}.
It is a challenge to design and analyze congestion avoiding algorithms with better ratios.

\section*{Acknowledgments}
Research supported by EU FET Integrated Project AEOLUS IST-15964, by FET EC 6th FP Research Project ARRIVAL FP6-021235-2, by the Dutch BSIK-BRICKS project, and by MIUR-FIRB Italy-Israel project RBIN047MH9.

{
\small
\bibliographystyle{abbrvnat}
\providecommand\SortNoop[1]{}

}

\end{document}